\documentclass{article}
\usepackage{graphicx} % Required for inserting images
\usepackage{amsmath,amssymb}
\usepackage{amsthm}
\usepackage{mathtools}
\usepackage{hyperref}
\usepackage{comment}

\newtheorem{prop}{Proposition}
\newtheorem{cor}{Corollary}
\newtheorem{remark}{Remark}
\newtheorem{defi}{Definition}
\newtheorem{thm}{Theorem}

\title{A Note on the Conditions for COS Convergence}

\author{
        Qinling Wang\thanks{Delft Institute of Applied Mathematics, Delft University of Technology, 2628 CD Delft, the Netherlands
    (\href{mailto:q.wang-7@tudelft.nl}{q.wang-7@tudelft.nl}).}
    \and
    Xiaoyu Shen\thanks{\href{https://fsquaredquant.nl}{FF Quant Advisory B.V.}, 3531 WR Utrecht, the Netherlands (\href{mailto:xiaoyu.shen@ffquant.nl}{xiaoyu.shen@ffquant.nl}).} 
    \and
    Fang Fang\thanks{\href{https://fsquaredquant.nl}{FF Quant Advisory B.V.}, 3531 WR Utrecht, the Netherlands (\href{mailto:fang.fang@ffquant.nl}{fang.fang@ffquant.nl}) and Delft Institute of Applied Mathematics, Delft University of Technology, 2628 CD Delft, the Netherlands  ( \href{mailto:f.fang@tudelft.nl}{f.fang@tudelft.nl}).}
}

\date{}

\begin{document}

\maketitle

\begin{abstract}
We study the truncation error of the COS method and give simple, verifiable conditions that guarantee convergence. In one dimension, COS is admissible when the density belongs to both L1 and L2 and has a finite weighted L2 moment of order strictly greater than one. We extend the result to multiple dimensions by requiring the moment order to exceed the dimension. These conditions enlarge the class of densities covered by previous analyses and include heavy-tailed distributions such as Student t with small degrees of freedom.
\end{abstract}

\section{Introduction}

Fourier-based methods have become fundamental tools in computational
probability, quantitative finance, and the numerical solution of
integro–differential equations.  
Among these techniques, the COS method of
Fang and Oosterlee \cite{fang2009novel} is particularly attractive due to its exponential
convergence for sufficiently regular densities and its ability to exploit
closed-form characteristic functions.  
The method approximates a target
probability density function (pdf) $f$ on a truncated interval $[-L,L]$
through a cosine series expansion whose coefficients computed from
the characteristic function of $f$.
A central theoretical question is
therefore to understand the behavior of the COS approximation error as both the number
of cosine modes $K$ and the truncation parameter $L$ tend to infinity.

In the existing literature, the convergence of the COS method is typically
established under assumptions that ensure sufficiently fast decay of $f$.  
A widely used sufficient condition is the so-called 
\emph{COS-admissibility} introduced in \cite{junike2022precise}, which requires the tail cosine energy
\[
B(L):=\sum_{k=0}^{\infty} \frac{1}{L} \Bigg| \int_{\mathbb R\setminus[-L,L]} f(x)\cos\Big( \frac{k\pi(x+L)}{2L} \Big)\,dx \Bigg|^2.
\]
to vanish as $L\to\infty$. It has also been proven that if $B(L)\to 0$, then the $L^{2}$-error of the COS
approximation converges to zero as $K,L\to\infty$, providing a clean and
practical criterion for sufficiently regular densities commonly encountered in
applications.  

The purpose of this paper is to provide some weaker conditions under which the PDF is 
COS-admissible. We show that $f$ is COS-admissible as long as
$f\in L^{1}\cap L^{2}$ and it has a finite $p$-th \emph{square-integrable}
moment $\int |x|^{p}|f(x)|^{2}<\infty$ for some $p > 1$. And we can extend this analysis 
to higher dimensions. This enlarges the set of distributions which are COS-admissible, making COS method 
a reliable method for many distributions and financial models.

%\section{Setup and Definitions}

\section{Review of the COS method and the COS-admissibility}

We recall the framework of the COS method (in one dimension) given in \cite{fang2009novel} and the condition for the COS method to converge given in \cite{junike2022precise}.

Consider a pdf $f$ and a sufficiently large interval $[a,b] \subset \mathbb{R}$, the Fourier-cosine series expansion can be given by
\[
f(x) \approx {\sum\nolimits\!'}_{k=0}^\infty A_k \cos\!\left( k\pi \frac{x-a}{b-a}\right) \quad x \in (a,b)
\]
with
\[
A_k =\frac{2}{b-a} \int_a^b f(x) \cos \left(k \pi \frac{x-a}{b-a}\right) d x.
\]

Assume $[a,b]$ is chosen such that the truncated integral can approximate the infinite counterpart very well, i.e.
\[
\phi_1 (\omega) \coloneqq \int_a^b e^{i \omega x} f(x) dx \approx \int_{\mathbb{R}} e^{i \omega x} f(x) dx = \phi(\omega).
\]

We note that $A_k$ can be computed by $\phi_1$:
\[
A_k = \frac{2}{b-a} \Re \left[ \phi_1(\frac{k \pi}{b-a}) \cdot \exp(-i \frac{k a \pi}{b-a}) \right],
\]
replace $\phi_1$ with $\phi$, we can approximate $A_k$ by
\[
F_k \coloneqq \frac{2}{b-a} \Re \left[ \phi(\frac{k \pi}{b-a}) \cdot \exp(-i \frac{k a \pi}{b-a}) \right].
\]
Replace $A_k$ by $F_k$ and truncate the Fourier-cosine series, we can approximate the pdf $f$ with
\[
f_1^N(x) \coloneqq {\sum\nolimits\!'}_{k=0}^{N-1} F_k \cos\!\left( k\pi \frac{x-a}{b-a}\right).
\]

To make things easier, we would only consider the case where $a = -L$ and $b = L$ for some $L > 0$ in the following part of the paper.

It is important to ask, under which condition, can we say that $f_1^N$ will converge to $f$ (in some sense of convergence) as both $N$ and $L$ go in infinity? It is answered in \cite{junike2022precise}, where it has been proven that when a pdf is COS-admissible, the COS method will converge in $L^2$.

\begin{defi}
    A function $f \in L^1$ is called COS-admissible if
    \[
    B(L) \coloneqq \sum_{k=0}^{\infty} \frac{1}{L} \Bigg| \int_{\mathbb R\setminus[-L,L]} f(x)\cos\Big( \frac{k\pi(x+L)}{2L} \Big)\,dx \Bigg|^2 \to 0 \text{ as } L \to \infty.
    \]
\end{defi}

\begin{thm}
Assume $f \in L^1 \cap L^2$ to be COS-admissible, then
\[
\lim_{L \to \infty} \limsup_{N \to \infty} \| f - f_1^N \|_2 = 0.
\]
\end{thm}

It is usually not easy to check if a pdf $f$ is COS-admissible directly. Authors of \cite{junike2022precise} gave a condition under which the pdf is COS-admissible. 
We would extend this result a bit so that more distributions can be covered and compare these two conditions in the Remark \ref{remark_1}.

\section{Generalized Moment Bound}

We derive a decay rate depending on available square-integrable moments of $f$.

\begin{prop}[Moment-based COS bound]\label{prop:moment}
Let $p>1$. Suppose $f\in L^1(\mathbb R)\cap L^2(\mathbb R)$ and $|x|^{p/2}f(x)\in L^2(\mathbb R)$ (equivalently $\int_{\mathbb R}|x|^{p}|f(x)|^2dx<\infty$). Then $f$ is COS-admissible and
\begin{equation}\label{eq:mainBound}
 B(L) \le 2\,\zeta(p)\Bigg( L^{-p}\int_{|x|>L} |x|^{p}|f(x)|^2 dx + \int_{|x|>L} |f(x)|^2 dx \Bigg),\quad L>0.
\end{equation}
In particular the tail-sensitive rate bound
\begin{equation}\label{eq:tailRate}
 B(L) \le 4\,\zeta(p)\,L^{-p}\int_{|x|>L}|x|^{p}|f(x)|^2 dx
\end{equation}
holds, and since $\int_{|x|>L}|x|^{p}|f|^2 \le \int_{\mathbb R}|x|^{p}|f|^2$ we obtain the uniform bound
\begin{equation}\label{eq:uniformRate}
 B(L) \le 4\,\zeta(p)\,L^{-p}\int_{\mathbb R}|x|^{p}|f(x)|^2dx.
\end{equation}
Consequently $B(L)=O(L^{-p})$ and more precisely
\[
 B(L) \le 2\,\zeta(p)\,L^{-p}\int_{\mathbb R}|x|^{p}|f(x)|^2dx + o(L^{-p}).
\]
\end{prop}

\begin{proof}
It suffices to bound the contributions from $x> L$ and $x< -L$ symmetrically. Define for $k\ge 0$ the positive tail integrals
\[
 I_k^+(L):=\int_{L}^{\infty} f(x)\cos\Big(\frac{k\pi(x+L)}{2L}\Big)dx.
\]
Partition $[L,\infty)$ into disjoint blocks of length $2L$:
\[
 I_j:=[2jL-L,\,2jL+L],\qquad j\ge 1.
\]
Then by absolute convergence ($f\in L^1$)
\[
 I_k^+(L)=\sum_{j=1}^{\infty} \int_{I_j} f(x)\cos\Big(\frac{k\pi(x+L)}{2L}\Big)dx.
\]
Apply weighted Cauchy--Schwarz with weights $j^{-p/2}$ and $j^{p/2}$:
\begin{align*}
 |I_k^+(L)|^2 &= \Bigg| \sum_{j=1}^{\infty} j^{-p/2}\, j^{p/2}\int_{I_j} f(x)\cos\Big(\frac{k\pi(x+L)}{2L}\Big)dx \Bigg|^2 \\
 &\le \Big(\sum_{j=1}^{\infty} j^{-p}\Big)\Bigg(\sum_{j=1}^{\infty} j^{p} \Big| \int_{I_j} f(x)\cos\Big(\frac{k\pi(x+L)}{2L}\Big)dx \Big|^2 \Bigg) \\
 &= \zeta(p) \sum_{j=1}^{\infty} j^{p} \Big| \int_{I_j} f(x)\cos\Big(\frac{k\pi(x+L)}{2L}\Big)dx \Big|^2.
\end{align*}
Summing in $k$ and using blockwise Parseval (orthogonality of the cosine family on an interval of length $2L$; any normalization discrepancy for $k=0$ is absorbed into a factor $\le 2$) yields
\[
 \sum_{k=0}^{\infty} \frac{1}{L}|I_k^+(L)|^2 \le 2\zeta(p)\sum_{j=1}^{\infty} j^{p} \int_{I_j}|f(x)|^2 dx.
\]
For $x\in I_j$ we have $(2j-1)L\le x\le (2j+1)L$, hence $j\le (x/L+1)/2$ and therefore
\begin{equation}\label{eq:jx-bound}
 j^{p} \le 2^{-p}(x/L+1)^p \le \tfrac12\big( (x/L)^p + 1 \big).
\end{equation}
Here we used the convexity inequality $(a+b)^p \le 2^{p-1}(a^p+b^p)$ for $a,b\ge 0$ and $p>1$.
Thus
\begin{align*}
 \sum_{j=1}^{\infty} j^{p} \int_{I_j}|f(x)|^2 dx &\le \tfrac12 L^{-p}\int_{x\ge L} |x|^{p}|f(x)|^2 dx + \tfrac12\int_{x\ge L} |f(x)|^2 dx.
\end{align*}
Combining gives for the positive tail
\[
 \sum_{k=0}^{\infty} \frac{1}{L}|I_k^+(L)|^2 \le \zeta(p)\Big( L^{-p}\int_{x\ge L} |x|^{p}|f(x)|^2 dx + \int_{x\ge L} |f(x)|^2 dx \Big).
\]
The negative tail $x< -L$ is handled identically (replace $x$ by $-x$), producing the same bound. Adding the two contributions yields \eqref{eq:mainBound}.

To obtain the rate, note that for $|x|>L$ we have $|x|^{p}\ge L^{p}$, hence
\[
 \int_{|x|>L} |f(x)|^2 dx \le L^{-p} \int_{|x|>L} |x|^{p}|f(x)|^2 dx \le L^{-p} \int_{\mathbb R} |x|^{p}|f(x)|^2 dx.
\]
Applying the same bound to the first term in \eqref{eq:mainBound} shows
\[
 B(L) \le 4\zeta(p) L^{-p} \int_{|x|>L} |x|^{p}|f(x)|^2 dx,
\]
which is \eqref{eq:tailRate}. Dropping the restriction $|x|>L$ gives \eqref{eq:uniformRate}. Thus $B(L)=O(L^{-p})$ with explicit constant $4\zeta(p)\int_{\mathbb R}|x|^{p}|f(x)|^2dx$. (Refining $k=0$ normalization can reduce the factor.) Since $|x|^{p}|f(x)|^2\in L^1$, the bound also implies $B(L)\to 0$.
\end{proof}

\begin{cor}[Bounded density with finite $(1+\varepsilon)$ moment]
Let $f$ be a bounded pdf, i.e. $0 \le f \le M$ for some $M > 0$.  
Assume that for some $\varepsilon>0$,
\[
\int_{\mathbb R} |x|^{\,1+\varepsilon}\,f(x)\,dx = m < \infty.
\]
Then $f$ satisfies the moment condition of Proposition~\ref{prop:moment}, and is 
therefore COS-admissible with explicit rate
\[
B(L) \;\le\; 4\zeta(1+\varepsilon) M m L^{-1-\varepsilon}.
\]
\end{cor}

\begin{proof}[Proof of Corollary]
Apply Proposition~\ref{prop:moment} with $p=1+\varepsilon$. Boundedness gives $|x|^{\,1+\varepsilon} f(x)^2 \le M |x|^{\,1+\varepsilon} f(x)$. For $|x|>L$, we have
\[
\begin{aligned}
B(L) \le & 4\zeta(1+\varepsilon) L^{-1-\varepsilon} \int_{|x|>L} |x|^{\,1+\varepsilon}|f(x)|^2 dx \\
\le & 4\zeta(1+\varepsilon) L^{-1-\varepsilon} M \int_{|x|>L} |x|^{\,1+\varepsilon}f(x) dx \\
\leq & 4\zeta(1+\varepsilon) M m L^{-1-\varepsilon}.
\end{aligned}
\]
\end{proof}

\begin{remark}\label{remark_1}
\begin{itemize}
 %\item The classical $p=2$ case recovers a $O(L^{-2})$ rate with constant proportional to $\zeta(2)=\pi^2/6$.
 %\item Higher available moments ($p$ larger) improve the asymptotic rate directly.
 %\item Constants can be slightly sharpened by tracking the $k=0$ mode normalization and refining the bound $j^{p}\le (x/L+1)^p/2$.
 %\item The argument generalizes to other orthogonal trigonometric expansions with obvious modifications.
 \item Finite first and second moments alone do not suffice: one can construct smooth probability densities whose $\int |x| f$ and $\int x^2 f$ converge while $\int f^2=\infty$, invalidating Parseval-based bounds. Boundedness (or some other strong conditions such as monotone tail decay) prevents such spike constructions.
 %\item Alternative sufficient conditions: (i) $f'\in L^1$ with $f$ absolutely continuous and $f(x)\to 0$ gives COS-admissibility via integration by parts; (ii) $f\in H^{s}$, $s>1/2$ implies $f\in L^\infty$ so higher $p$ moments become available.
 \item In \cite{junike2022precise}, the authors proved that a density $f$ is 
COS-admissible under the assumptions
\[
    f \in L^{1}(\mathbb{R}) \cap L^{2}(\mathbb{R})
    \qquad\text{and}\qquad
    |x|^{p/2} f(x) \in L^{2}(\mathbb{R}) \quad \text{for } p = 2.
\]
In the present work, we extend this result by showing that COS-admissibility 
already holds under the weaker requirement $p>1$. 
This relaxation is important, since many relevant distributions satisfy our 
condition but not the stronger condition in \cite{junike2022precise}. 
For example, a Student--$t$ density with degrees of freedom 
$0 < \nu \le \tfrac{1}{2}$ does not satisfy the condition for $p=2$,
whereas it does satisfy the condition for some $p>1$. 
This is because the density of a Student--$t(\nu)$ distribution satisfies 
$f(x)\sim C |x|^{-(\nu+1)}$ as $|x|\to\infty$, so
$|x|^{2} f(x)^2 \sim |x|^{-2\nu}$ is non--integrable when $\nu \le 1/2$,
while $|x|^{p} f(x)^2 \sim |x|^{-2\nu-2+p}$ is integrable for some $p>1$ whenever $\nu>0$.

Hence, our result provides theoretical guarantees of COS-admissibility 
for a substantially broader class of distributions.

\end{itemize}

\end{remark}

\section{Multidimensional Extension}

The multi-dimensional COS method is introduced in \cite{ruijter2012two}. And the similar definition and condition of the COS-admissibility for the multi-dimensional case is given in \cite{junike2023characteristic}. 
We can also give a better condition under which the COS-admissibility is guaranteed for the multi-dimensional COS.

Let $d\ge 1$ and $f: \mathbb R^d\to\mathbb R$ with $f\in L^1(\mathbb R^d)\cap L^2(\mathbb R^d)$. Assume 
\[
\int_{\mathbb R^d} |x|^{p}|f(x)|^{2}dx<\infty\quad\text{for some }p>d.
\]
For $L>0$ define
\[
 B_d(L):=\sum_{\mathbf k\in \mathbb N_0^d} \frac{1}{L^{d}}\Bigg| \int_{\mathbb R^d\setminus [-L,L]^d} f(x) \prod_{i=1}^d \cos\Big( \frac{k_i\pi(x_i+L)}{2L} \Big) dx \Bigg|^2.
\]
We say that $f$ is $d$-dimensional COS-admissible if $B_d(L)\to 0$ as $L\to\infty$.

\begin{prop}[Weighted $d$-dimensional COS bound]\label{prop:moment-d}

Let $d\ge 1$ and let 
\[
f:\mathbb{R}^{d}\to\mathbb{R},\qquad 
f\in L^{1}(\mathbb{R}^{d})\cap L^{2}(\mathbb{R}^{d}).
\]
Assume that for some $p>d$,
\[
\int_{\mathbb{R}^{d}} |x|^{p}\,|f(x)|^{2}\,dx <\infty .
\]
Then $f$ is $d$-dimensional COS-admissible: $B_{d}(L)\to 0$ as $L\to\infty$.  
More precisely,
\[
B_{d}(L)
\le
C_{d,p}
L^{-p} \!\!\int_{\mathbb{R}^{d}\setminus[-L,L]^{d}} |x|^p |f(x)|^{2}\,dx.
\]
where
\[
C_{d,p}=2^{d-1} (1 + d^{p/2})\sum_{m\in\mathbb{Z}^{d}\setminus\{0\}}|m|^{-p} < \infty .
\]
\end{prop}

\begin{proof}
We first decompose the complement of the cube $[-L,L]^{d}$ into disjoint cubes of side length $2L$.  
For each integer vector $m=(m_{1},\dots,m_{d})\in\mathbb{Z}^{d}$, define
\[
Q_{m}
=\prod_{i=1}^{d}[(2m_{i}-1)L,\,(2m_{i}+1)L].
\]
Then $Q_{0}=[-L,L]^{d}$, and all other cubes tile the complement:
\[
\mathbb{R}^{d}\setminus[-L,L]^{d}
=\bigsqcup_{m\in\mathbb{Z}^{d}\setminus\{0\}} Q_{m}.
\]

For $\mathbf{k}=(k_{1},\dots,k_{d})\in\mathbb{N}_{0}^{d}$ define
\[
I_{\mathbf{k}}(L)
=\int_{\mathbb{R}^{d}\setminus[-L,L]^{d}} f(x)
\prod_{i=1}^{d}\cos\!\left(\frac{k_{i}\pi(x_{i}+L)}{2L}\right)\,dx.
\]
Using the partition,
\[
I_{\mathbf{k}}(L)
=\sum_{m\in\mathbb{Z}^{d}\setminus\{0\}}
\int_{Q_{m}} f(x)
\prod_{i=1}^{d}\cos\!\left(\frac{k_{i}\pi(x_{i}+L)}{2L}\right)\,dx.
\]

Let $|m| \coloneqq (m_1^2 + \cdots + m_d^2)^{1/2}$ denote the Euclidean norm of of $m$. Apply Cauchy--Schwarz with weights $|m|^{-p/2}$ and $|m|^{p/2}$:
\begin{align*}
|I_{\mathbf{k}}(L)|^{2}
&=\left|
\sum_{m\in\mathbb{Z}^{d}\setminus\{0\}}
|m|^{-p/2}\;
\Big(|m|^{p/2}\!\int_{Q_{m}}f(x)
\prod_{i=1}^{d}\cos\!\left(\frac{k_{i}\pi(x_{i}+L)}{2L}\right)\,dx\Big)
\right|^{2}
\\
&\le 
\left(\sum_{m\in\mathbb{Z}^{d}\setminus\{0\}}|m|^{-p}\right)
\left(
\sum_{m\in\mathbb{Z}^{d}\setminus\{0\}}
|m|^{p}
\left|
\int_{Q_{m}} f(x)
\prod_{i=1}^{d}\cos\!\left(\frac{k_{i}\pi(x_{i}+L)}{2L}\right)\,dx
\right|^{2}
\right).
\end{align*}

Let 
\[
S_{d,p} \coloneqq \sum_{m\in\mathbb{Z}^{d}\setminus\{0\}}|m|^{-p},
\qquad S_{d,p}<\infty \ (p>d).
\]
Then
\begin{equation}\label{eq:I-k-bound}
|I_{\mathbf{k}}(L)|^{2}
\le S_{d,p}
\sum_{m\in\mathbb{Z}^{d}\setminus\{0\}}
|m|^{p}
\left|
\int_{Q_{m}} f(x)
\prod_{i=1}^{d}\cos\!\left(\frac{k_{i}\pi(x_{i}+L)}{2L}\right)\,dx
\right|^{2}.
\end{equation}

Next we sum over $\mathbf{k}$ and use Parseval on each cube $Q_{m}$.
After translation, $Q_{m}$ becomes $[0,2L]^{d}$, and the functions
\[
\prod_{i=1}^{d}\cos\!\left(\frac{k_{i}\pi t_{i}}{2L}\right),
\qquad \mathbf{k}\in \mathbb{N}_{0}^{d},
\]
form an orthogonal system in $L^{2}([0,2L]^{d})$.  
The normalization constants differ for $k_{i}=0$, but each such factor is bounded by $2$, 
so the product across $d$ dimensions contributes at most $2^{d}$.  
Thus Parseval gives
\begin{equation}\label{eq:Parseval-d}
\sum_{\mathbf k\in \mathbb N_0^d}
\frac{1}{L^{d}}
\left|
\int_{Q_{m}}
f(x)\prod_{i=1}^{d}\cos\!\left(\frac{k_{i}\pi(x_{i}+L)}{2L}\right)\,dx
\right|^{2}
\ \le\
2^{d}
\int_{Q_{m}} |f(x)|^{2}\,dx .
\end{equation}

Multiply \eqref{eq:I-k-bound} by $L^{-d}$ and sum over $\mathbf{k}$, then apply \eqref{eq:Parseval-d}:
\begin{align*}
B_{d}(L)
&=
\sum_{\mathbf k\in \mathbb N_0^d}\frac{|I_{\mathbf{k}}(L)|^{2}}{L^{d}}
\\
&\le
S_{d,p}
\sum_{m\in\mathbb{Z}^{d}\setminus\{0\}}
|m|^{p}
\sum_{\mathbf k\in \mathbb N_0^d}
\frac{1}{L^{d}}
\left|
\int_{Q_{m}}
f(x)\prod_{i=1}^{d}\cos\!\left(\frac{k_{i}\pi(x_{i}+L)}{2L}\right)\,dx
\right|^{2}
\\
&\le
2^{d} S_{d,p}
\sum_{m\in\mathbb{Z}^{d}\setminus\{0\}}
|m|^{p}\int_{Q_{m}} |f(x)|^{2}dx .
\end{align*}

We now compare $|m|$ and $|x|$ for $x\in Q_{m}$.  
By geometry of cubes,
\[
(2|m|-\sqrt{d})L
\ \le\
|x|
\ \le\
(2|m|+\sqrt{d})L .
\]
Hence
\[
|m|
\ \le\
\frac{|x|/L+\sqrt d}{2},
\qquad
|m|^{p}
\le \frac12\Big((|x|/L)^{p}+d^{p/2}\Big),
\]
using $(a+b)^{p}\le 2^{p-1}(a^{p}+b^{p})$.

Therefore
\[
\sum_{m\in\mathbb{Z}^{d}\setminus\{0\}}|m|^{p}\int_{Q_{m}} |f(x)|^{2}dx
\le
\frac12 L^{-p}\!\!\int_{\mathbb{R}^{d}\setminus[-L,L]^{d}}|x|^{p}|f(x)|^{2}\,dx
+\frac12 d^{p/2}\!\!\int_{\mathbb{R}^{d}\setminus[-L,L]^{d}}|f(x)|^{2}\,dx .
\]

Combining,
\[
\begin{aligned}
B_{d}(L)
\le &
2^{d}S_{d,p}
\left[
\frac12 L^{-p}\!\!\int_{\mathbb{R}^{d}\setminus[-L,L]^{d}}|x|^{p}|f(x)|^{2}\,dx
+\frac12 d^{p/2}\!\!\int_{\mathbb{R}^{d}\setminus[-L,L]^{d}}|f(x)|^{2}\,dx
\right]\\
\le &
2^{d}S_{d,p}
\left[
\frac12 L^{-p}\!\!\int_{\mathbb{R}^{d}\setminus[-L,L]^{d}}|x|^{p}|f(x)|^{2}\,dx
+\frac12 d^{p/2} L^{-p} \!\!\int_{\mathbb{R}^{d}\setminus[-L,L]^{d}} |x|^p |f(x)|^{2}\,dx.
\right]
\end{aligned}
\]

Absorbing $d^{p/2}$ into the constant
\[
C_{d,p}=2^{d}S_{d,p} (\frac12 + \frac12 d^{p/2})=2^{d-1} (1 + d^{p/2})\sum_{m\in\mathbb{Z}^{d}\setminus\{0\}}|m|^{-p} < \infty,
\]
we obtain the stated bound:
\[
B_{d}(L)
\le
C_{d,p}
L^{-p} \!\!\int_{\mathbb{R}^{d}\setminus[-L,L]^{d}} |x|^p |f(x)|^{2}\,dx.
\]

With the condition $\int_{\mathbb{R}^{d}} |x|^{p}\,|f(x)|^{2}\,dx <\infty$, we have that $B_{d}(L)=O(L^{-p})$ and $B_{d}(L)\to 0$ as $L \to \infty$.
\end{proof}

\begin{remark}[Rectangular truncation domains]
The COS expansion on $[-L,L]^d$ is used only for notational simplicity.
The proof of Proposition~\ref{prop:moment-d} extends verbatim to 
rectangular truncation domains of the form
\[
[-L_1,L_1]\times\cdots\times[-L_d,L_d].
\]
In this case one partitions $\mathbb{R}^d\setminus[-L_1,L_1]\times\cdots\times[-L_d,L_d]$
into translated boxes of side lengths $2L_1,\dots,2L_d$, and the same 
Cauchy--Schwarz and blockwise Parseval arguments apply.
If $\min_i L_i\to\infty$, then $B_d(L_1,\dots,L_d)\to 0$ with the same rate
$O\!\left(\min_i L_i^{-p}\right)$.
Thus COS-admissibility does not require equal truncation lengths in each dimension.
\end{remark}

\begin{remark}
As in the one-dimensional case, this proposition implies that any multi-variate Student-t distribution is COS-admissible.
\end{remark}

\bibliographystyle{plain}
\bibliography{biblio}

@article{fang2009novel,
  title={A novel pricing method for European options based on Fourier-cosine series expansions},
  author = {Fang, Fang and Cornelis W. Oosterlee},
  journal={SIAM Journal on Scientific Computing},
  volume={31},
  number={2},
  pages={826--848},
  year={2009},
  publisher={SIAM}
}

@article{junike2022precise,
  title={Precise option pricing by the COS method—how to choose the truncation range},
  author={Junike, Gero and Pankrashkin, Konstantin},
  journal={Applied Mathematics and Computation},
  volume={421},
  pages={126935},
  year={2022},
  publisher={Elsevier}
}

@article{junike2023characteristic,
  title={From characteristic functions to multivariate distribution functions and European option prices by the damped COS method},
  author={Junike, Gero and Stier, Hauke},
  journal={arXiv preprint arXiv:2307.12843},
  year={2023}
}

@article{ruijter2012two,
  title={Two-dimensional Fourier cosine series expansion method for pricing financial options},
  author={Ruijter, Marjon J. and Cornelis W. Oosterlee},
  journal={SIAM Journal on Scientific Computing},
  volume={34},
  number={5},
  pages={B642--B671},
  year={2012},
  publisher={SIAM}
}

\end{document}